\documentclass[aps,11pt,twoside, nofootinbib]{revtex4}
\usepackage{amsthm}
\usepackage{amsmath,latexsym,amssymb,verbatim,enumerate,graphicx}
\usepackage{color}

\usepackage{amsmath,amsfonts,amssymb,graphics,graphicx, color,times,bbm}
\usepackage[sort&compress]{natbib}

\newcommand{\me}{\mathrm{e}}

\newcommand{\md}{\mathrm{d}}

\newcommand{\nn}{\mathbbm{N}}
\newcommand{\rr}{\mathbbm{R}}

\newtheorem{definition}{Definition} 
\newtheorem{prop}[definition]{Proposition}
\newtheorem{lemma}[definition]{Lemma}

\newtheorem{thm}[definition]{Theorem}
\newtheorem{corollary}[definition]{Corollary}

\newtheorem*{rep@theorem}{\rep@title}
\newcommand{\newreptheorem}[2]{%
\newenvironment{rep#1}[1]{%
 \def\rep@title{#2 \ref{##1} (restatement)}%
 \begin{rep@theorem}}%
 {\end{rep@theorem}}}
\makeatother

\newreptheorem{thm}{Theorem}
\newreptheorem{lem}{Lemma}

\def\ba#1\ea{\begin{align}#1\end{align}}
\def\ban#1\ean{\begin{align*}#1\end{align*}}

\newcommand{\be}{\begin{equation}}
\newcommand{\ee}{\end{equation}}

\def\benum{\begin{enumerate}}
\def\eenum{\end{enumerate}}

\def\squareforqed{\hbox{\rlap{$\sqcap$}$\sqcup$}}
\def\qed{\ifmmode\squareforqed\else{\unskip\nobreak\hfil
\penalty50\hskip1em\null\nobreak\hfil\squareforqed
\parfillskip=0pt\finalhyphendemerits=0\endgraf}\fi}
\def\endenv{\ifmmode\;\else{\unskip\nobreak\hfil
\penalty50\hskip1em\null\nobreak\hfil\;
\parfillskip=0pt\finalhyphendemerits=0\endgraf}\fi}

\newcommand{\tr}{\text{tr}}

\newcommand{\id}{\mathbbm{1}}

\newcommand{\<}{\langle}
\renewcommand{\>}{\rangle}

\def\be{\begin{equation}}
\def\ee{\end{equation}}
\def\ben{\begin{eqnarray}}
\def\een{\end{eqnarray}}

\def\bei{\begin{itemize}}
\def\eei{\end{itemize}}

\mathchardef\ordinarycolon\mathcode`\:
\mathcode`\:=\string"8000
\def\vcentcolon{\mathrel{\mathop\ordinarycolon}}
\begingroup \catcode`\:=\active
  \lowercase{\endgroup
  \let :\vcentcolon
  }

\newcommand{\nc}{\newcommand}
 \nc{\proj}[1]{|#1\rangle\!\langle #1 |} 
\nc{\avg}[1]{\langle#1\rangle}

\nc{\conv}{\operatorname{conv}}
\nc{\smfrac}[2]{\mbox{$\frac{#1}{#2}$}} \nc{\Tr}{\operatorname{Tr}}
\nc{\ox}{\otimes} \nc{\dg}{\dagger} \nc{\dn}{\downarrow}
\nc{\lmax}{\lambda_{\text{max}}}
\nc{\lmin}{\lambda_{\text{min}}}

\nc{\csupp}{{\operatorname{csupp}}}
\nc{\qsupp}{{\operatorname{qsupp}}} \nc{\var}{\operatorname{var}}
\nc{\rar}{\rightarrow} \nc{\lrar}{\longrightarrow}
\nc{\poly}{\operatorname{poly}}
\nc{\polylog}{\operatorname{polylog}} \nc{\Lip}{\operatorname{Lip}}
\nc{\Om}{\Omega}
\nc{\wt}[1]{\widetilde{#1}}

\def\>{\rangle}
\def\<{\langle}

\nc{\glneq}{{\raisebox{0.6ex}{$>$}  \hspace*{-1.8ex} \raisebox{-0.6ex}{$<$}}}
\nc{\gleq}{{\raisebox{0.6ex}{$\geq$}\hspace*{-1.8ex} \raisebox{-0.6ex}{$\leq$}}}

\nc{\vholder}[1]{\rule{0pt}{#1}}
\nc{\wh}[1]{\widehat{#1}}
\nc{\h}[1]{\widehat{#1}}

\nc{\ob}[1]{#1}

\def\beq{\begin {equation}}
\def\eeq{\end {equation}}

\def\be{\begin{equation}}
\def\ee{\end{equation}}

\nc{\eq}[1]{(\ref{eq:#1})} 
\nc{\eqs}[2]{\eq{#1} and \eq{#2}}

\nc{\eqn}[1]{Eq.~(\ref{eqn:#1})}
\nc{\eqns}[2]{Eqs.~(\ref{eqn:#1}) and (\ref{eqn:#2})}

\nc{\region}{\cS\cW}

\begin{document}

\title{{\Large Equivalence of Statistical Mechanical Ensembles for Non-Critical Quantum Systems}}

\author{Fernando G.S.L. Brand\~ao}
\email{f.brandao@ucl.ac.uk}
\affiliation{Quantum Architectures and Computation Group, Microsoft Research, Redmond, WA}
\affiliation{Department of Computer Science, University College London WC1E 6BT, United Kingdom}

\author{Marcus Cramer}
 \email{marcus.cramer@uni-ulm.de}
\affiliation{Institut f\"ur Theoretische Physik, Universit\"at Ulm, Germany}

\begin{abstract}
 
We consider the problem of whether the canonical and microcanonical ensembles are locally equivalent for short-ranged quantum Hamiltonians of $N$ spins arranged on a $d$-dimensional lattices. For any  temperature for which the system has a finite correlation length, we prove that the canonical and microcanonical state
are approximately equal on regions containing up to $O(N^{1/(d+1)})$ spins. The proof rests on a variant of the Berry--Esseen theorem for quantum lattice systems and ideas from quantum information theory. 

\end{abstract}

\maketitle

\parskip .75ex


\section{Introduction}
In statistical mechanics there are two main ensembles (at zero chemical potential) that can be used to compute equilibrium properties of large systems: the microcanonical and canonical ensembles. Roughly, the first describes the physics of a system that is isolated and has total fixed energy. The second describes the physics of a system that is at thermal equilibrium with a large environment at fixed temperature. It turns out that in many cases, although not all, the two ensembles give the same predictions for very large systems. There is a long sequence of studies aiming at elucidating under what conditions the two ensembles can be used interchangeably (see e.g. \cite{LL69, Lima1, Lima2, Tou09, Geo95, MAMW13} and the discussion below). 

In textbooks the canonical ensemble is commonly introduced by considering the microcanonical ensemble of the system and a large environment and restricting to observables acting on the system only. Under the assumption that the interactions of the system and environment are very weak, the canonical ensemble can be derived. However in many situations the assumption of weak coupling is not justified. For example, in many closed quantum systems small regions thermalize \cite{PSSV11,EFG14}; in this case the remaining of the system is acting as an environment. It is therefore an interesting problem to find more general conditions that guarantee the equivalence of the two ensembles. Our main goal is to give one such condition: We show that short ranged interactions and a finite correlation length lead to the equivalence of ensembles for every sufficiently large finite volume. The condition of a finite correlation length (and short ranged interactions) is known to be required (see e.g. \cite{Des}).

\section{Results}

We let $\Lambda = \{1, \ldots, n \}^d$ be a finite collection of {\it vertices} or {\it lattice sites} in $d$ dimensions with $N = |\Lambda|=n^d$ sites. 
We consider local Hamiltonians, acting on the Hilbert space $\mathcal{H}=\otimes_{i\in\Lambda}\mathcal{H}_i$, $\dim \mathcal{H}_i= D$, given by
\begin{equation}
\label{ham}
H = \sum_{i \in \Lambda} H_i=\sum_\nu E_\nu|\nu\rangle\langle \nu|,
\end{equation}
where we assume that the $H_i$ are bounded, $\|H_i\| \le 1$, and local in the sense that $H_i$ acts only on sites $j$ with $\text{dist}(i,j)\le k$ (for the Manhattan metric $\text{dist}(.,.)$ in the lattice). 

For such $k$-local Hamiltonians, we let $\rho_{T} := e^{-H / T} / Z(T)$ be the canonical state at temperature $T$ (also known as Gibbs state or thermal state) and $Z(T) := \tr(e^{-H / T} )$ the partition function (we set Boltzmann's constant to unit). In the canonical ensemble at temperature $T$, averages are computed using $\rho_T$. The energy density at temperature $T$ is given by
\begin{equation}
u(T) := \frac{1}{N} \tr(H \rho_T),
\end{equation}
the specific heat capacity at temperature $T$ by
\begin{equation}
c(T) := \frac{d u(T')}{d T'} \bigg|_{T' = T} =  \frac{1}{N T^2} \left(  \tr[H^2 \rho_{T}] - (\tr[H \rho_{T}])^2 \right),
\end{equation}
and the entropy density by\footnote{Throughout, we denote by $\ln$ ($\log$) the logarithm to the base $\me$ ($2$).}
\begin{equation}
s(T) := - \frac{1}{N} \tr[\rho_T\ln(\rho_T)].
\end{equation}

Given regions $X,Y\subset\Lambda$, we denote by $\tr_{\Lambda \backslash X}$ the partial trace over the complement of $X$ in $\Lambda$
and for states $\rho \in {\cal D}(\mathcal{H})$ (the set of density matrices acting on $\mathcal{H}$), we denote $\rho_{XY}=\tr_{\Lambda \backslash (X\cup Y)}(\rho)$. Given two states $\rho,\sigma$, their trace-norm distance is 
\begin{equation}
\Vert \rho - \sigma \Vert_1 :=  \tr(|\rho - \sigma|) = \max_{0 \leq M \leq I} 2\tr(M(\rho - \sigma))
\end{equation}
and quantifies how distinguishable the two states are.

We say a state $\rho \in {\cal D}(\mathcal{H})$  has $(\xi,z)$-exponentially decaying correlations (or a $(\xi,z)$-finite correlation length) if
there are $\xi>0$ and $z\ge 0$ such that for every two regions $X, Y\subset\Lambda$ with $\text{dist}(X,Y)> 0$,
\begin{eqnarray} \label{decaycor}
\text{cor}_{\rho}(X, Y) := \max_{\substack{P, \hspace{0.05 cm} Q \\ \text{supp}(P)\subset X\\ \text{supp}(Q) \subset Y}} \frac{|\tr((P \otimes Q)(\rho_{XY} - \rho_{X}\otimes \rho_{Y}))|}{\Vert P \Vert  \Vert Q \Vert } \leq N^z\me^{- \text{dist}(X, Y) / \xi},
\end{eqnarray}
where
\begin{equation}
\text{dist}(X, Y) := \min_{x \in X, y \in Y} \text{dist}(x, y).
\end{equation}

Given $e\in\rr$ and $\delta > 0$, let
\begin{equation}
M_{e, \delta} := \{ \nu :  | E_\nu - e N   | \leq \delta \sqrt{N}  \},
\end{equation}
and define the microcanonical state of mean energy $e$ and energy spread $\delta \sqrt{N}$ by
\begin{equation}
\tau_{e, \delta} := \frac{1}{|M_{e, \delta}|} \sum_{\nu \in M_{e, \delta}} |\nu \rangle \langle \nu |.
\end{equation} 
In the microcanonical ensemble averages are computed using $\tau_{e, \delta}$.

Finally, for a Gibbs state corresponding to a Hamiltonian as in Eq.~\eqref{ham} and with $(\xi,z)$-exponentially decaying correlations, we define (see Lemma~\ref{thmBerryEsseenThm}):
\begin{equation}
\begin{split}
\Delta_{k,\xi,z,T}&:=C_d \frac{(\max\{k,\xi\}(z+1))^{2d}}{\sqrt{T^2c(T)}}\max\left\{\frac{1}{\max\{k,\xi\}(z+1)\ln(N)},\frac{1}{T^2c(T)}\right\},
\end{split}
\end{equation}
where $C_d\ge 1$ is a constant which only depends on the dimension of the lattice $\Lambda$.

We can now state the main result. It shows that for general quantum many-body systems at non-critical temperatures (meaning that the canonical state has a finite correlation length), the canonical ensemble gives essentially the same predictions as the microcanonical ensemble, for every observable which acts on sufficiently small regions.

\vspace{0.2 cm}

\begin{thm}  \label{equivalenceensembles} Let $\mathcal{C}_l$ be the set of all hypercubes contained in $\Lambda=\{1,\dots,n\}^d$  with edge length $l\in\nn$, $1\le l\le\frac{n+1}{2}$, and let $N=n^d>2$. 
Let the canonical state $\rho_T$ (corresponding to a $k$-local Hamiltonian as in Eq.~\eqref{ham}) with energy density $u(T)$ and specific heat capacity $c(T)$ have $(\xi,z)$-exponentially decaying correlations. Let
the microcanonical state $\tau_{e,\delta}$ have mean energy $e$ such that
\begin{equation} \label{eq1}
 |e-u(T)|\le  \sqrt{c(T)T^2/N}
 \end{equation}
  and energy spread $\delta \sqrt{N}$ such that
\begin{equation} \label{eq2}
28\Delta_{k,\xi,z,T}\sqrt{c(T)T^2}\frac{\ln^{2d}(N)}{\sqrt{N}}\le \delta\le \sqrt{c(T)T^2}.
\end{equation}
Let $\epsilon>0$. If
\begin{equation} \label{complicatedbound}
\frac{56\sqrt{c(T)}\Delta_{k,\xi,z,T}\ln^{2d}(N)
+(5+\epsilon z)\ln(N)}{\epsilon \ln(2)}+
\frac{2\xi\ln(D)l^d+l+2}{\xi \ln(2) }\le \left(\frac{ \epsilon N}{\ln^d(4)\xi^d}\right)^{\frac{1}{d+1}}
\end{equation}
then
\begin{equation}
\mathop{\mathbb{E}}_{C \in {\cal C}_l}\|(\tau_{e,\delta})_{C}-(\rho_T)_{C}\|_1 \le 7\sqrt{\epsilon},
\end{equation}
where the expectation is taken uniformly over $\mathcal{C}_l$.
\end{thm}

We note the following: 
\begin{enumerate}

\item Eq.~(\ref{complicatedbound}) is satisfied whenever $N$ is sufficiently large and $l^d \leq O( N^{1/(d+1)})$.

\item We do not need to take the average over regions $C \in {\cal C}_l$ if we assume the Hamiltonian is translation invariant.

\item The condition of a finite correlation length is necessary. Indeed the two ensembles differ in the Ising model approaching the critical point, when the correlation length diverges, for regions of size $O(\log(N))$ (see e.g. \cite{Des}). It is an open question if a similar result can be obtained for critical systems and small enough regions, assuming that correlations decay algebraically; in our proof it is important that the correlations decay exponentially as in Eq.~(\ref{decaycor}). 

\item 
Any system is expected to have a finite correlation length whenever it is away from a critical point. One can rigorously show that one-dimensional systems always have a finite correlation length at any temperature \cite{Ara69}, while in any dimension there is a critical temperature (depending only on the geometry of the lattice) above which every system has a finite correlation length \cite{KGKRE14}.

\end{enumerate}

An important step in the proof of the theorem will be to establish the following proposition, which we believe is of independent interest and which we prove in a stronger version (Proposition~\ref{strongerProp}) in Section~\ref{proof:prop1}. It shows that two states $\tau$ and $\rho$ are locally equivalent whenever their quantum relative entropy
\begin{equation}
S(\tau \| \rho) = \tr(\tau (\log \tau - \log \rho))
\end{equation}
is $O(N^{1/(d+1)})$ and $\rho$ has finite correlation length.

\begin{prop} \label{weakerProp} Let $\mathcal{C}_l$ be the set of all hypercubes contained in $\Lambda=\{1,\dots,n\}^d$  with edge length $l\in\nn$, $1\le l\le\frac{n+1}{2}$, and let $N=n^d>1$. Let $\epsilon>0$, the states $\rho$, $\tau$, and $l$ such that $\rho$ has $(\xi,z)$-exponentially decaying correlations and such that 
\begin{equation}
\label{simple:condition}
 \frac{S(\tau\|\rho)+3}{\epsilon }+
\frac{2\xi\ln(D)l^d+l+2}{\xi \ln(2) }+\log(N^{z+1})\le \left(\frac{ \epsilon N}{\ln^d(4)\xi^d}\right)^{\frac{1}{d+1}}.
\end{equation}
Then
\begin{equation}
\label{simple:result}
\mathop{\mathbb{E}}_{C \in {\cal C}_l}\|\tau_{C}-\rho_{C}\|_1 \le 7\sqrt{\epsilon},
\end{equation}
where the expectation is taken uniformly over $\mathcal{C}_l$.
\end{prop}

We note that the state $\rho$ does not need to be a thermal state, we merely demand it to have $(\xi,z)$-exponentially decaying correlations.
If we do assume it is a thermal state then, as $TS(\tau\|\rho_T)=F_T(\tau)-F_T(\rho_T)$, where, given a Hamiltonian $H$ and temperature $T$, the free energy of a state $\tau$ is given by 
\begin{equation}
F_T(\tau) = \tr(H \tau) - T S(\tau),
\end{equation}
the proposition shows that for temperatures away from criticality, any state of small free energy must have approximately thermal averages for local observables. Theorem \ref{equivalenceensembles} will then follow from Proposition \ref{weakerProp} and showing that $\tau_{e, \delta}$ has small free energy whenever $\rho_{T}$ has a finite correlation length.

\subsection{Beyond Microcanonical States}
How crucial is the use of the microcanonical ensemble? Proposition \ref{weakerProp} shows that not only the microcanonical state, but any state of small enough free energy is locally thermal. It turns out that Theorem \ref{equivalenceensembles} can be extended in two additional respects:
(1) It applies to any state that lives in the microcanonical subspace and has sufficiently large entropy.
(2) Following \cite{PSW06, GLTZ06}, it applies to a generic state in the microcanonical subspace, with overwhelming probability with respect to the Haar measure.

\begin{corollary}\label{corollaryProp} Let $\mathcal{C}_l$ be the set of all hypercubes contained in $\Lambda=\{1,\dots,n\}^d$  with edge length $l\in\nn$, $1\le l\le\frac{n+1}{2}$, and let $N=n^d>2$. 
Let the canonical state $\rho_T$ (corresponding to a $k$-local Hamiltonian as above), with energy density $u(T)$ and specific heat capacity $c(T)$, have $(\xi,z)$-exponentially decaying correlations. Let  $M_{e,\delta}$ with $e,\delta$ such that Eqs. (\ref{eq1}) and (\ref{eq2}) hold true. Let $\epsilon>0$ and 
\begin{equation}
\frac{56\sqrt{c(T)}\Delta_{k,\xi,z,T}\ln^{2d}(N)
+(5+\epsilon z)\ln(N)}{\epsilon \ln(2)}+
\frac{2\xi\ln(D)l^d+l+2}{\xi \ln(2) }\le \frac{1}{2}\left(\frac{ \epsilon N}{\ln^d(4)\xi^d}\right)^{\frac{1}{d+1}}.
\end{equation}

\begin{enumerate}
\item Let $\tau$ a state on the subspace spanned by $\{|\nu\rangle\}_{\nu\in M_{e,\delta}}$ with entropy 
\begin{equation} 
S(\tau)
\ge \log(|M_{e, \delta}|)-\frac{\epsilon}{2}\left(\frac{ \epsilon N}{\ln^d(4)\xi^d}\right)^{\frac{1}{d+1}}.
\end{equation} 
Then $\mathop{\mathbb{E}}_{C \in {\cal C}_l}\|\tau_{C}-(\rho_T)_{C}\|_1 \le 7\sqrt{\epsilon}$.
\item Let $\tau$ be a pure state drawn from the Haar measure on $\text{span}\{|\nu\rangle\}_{\nu\in M_{e,\delta}}$. Then, with probability at least $1-2\me^{-1/\eta}$,
\begin{equation}
\mathop{\mathbb{E}}_{C \in {\cal C}_l}\|\tau_{C}-(\rho_T)_{C}\|_1 \le7\sqrt{\epsilon}+
 \eta+\frac{D^{l^d}}{\sqrt{18\pi^3}}\eta^{3/2},
\end{equation}
with
\begin{equation} 
\eta:=18^{1/3}\pi\exp\left[-\frac{N}{3}\left(s(T)-\frac{2\sqrt{c(T)}+2}{\sqrt{N}}\right)\right].
\end{equation}
\end{enumerate}
Here, the expectation is taken uniformly over $\mathcal{C}_l$.
\end{corollary}

The second part of the corollary is a direct consequence of the first part, the quantum Berry--Esseen bound in Lemma~\ref{thmBerryEsseenThm}, and the result of \cite{PSW06, GLTZ06}, which shows that a generic state in a energy subspace has the same local reductions as the microcanonical state.

\section{Comparison with Previous Work} 

The problem of equivalence of ensembles has been considered since the foundational work of Boltzmann and Gibbs. See \cite{Tou06} for a historical perspective. An intuitive explanation for the equivalence at non-critical temperatures is the following: Whenever there is a finite correlation length, the heat capacity, which determines the energy variance density, is of order $O(N)$, implying that the standard deviation of energy is of order $O(\sqrt{N})$. As energy is an extensive quantity (i.e. of order $O(N)$), the distribution of energies in the Gibbs state is highly concentrated around its mean value for a large number of spins, and so one might be tempted to believe that only the microcanonical subspace is relevant. However it turns out that this argument is too simplistic. Indeed it is easy to see that for any $\delta > 0$, $\tau_{e, \delta}$ and $\rho_T$ (with $e = u(T)$) are nearly orthogonal for sufficiently large $N$. Therefore any meaningful argument for the equivalence of ensembles must go beyond the distribution of energies and in some way restrict the kind of observables considered (for example, considering observables acting in small regions).

The most fruitful direction explored so far has been to consider systems in the thermodynamical limit. In this regime one can prove the equivalence of ensembles on the level of thermodynamical functions \cite{LL69, Lima1, Lima2, Tou09} (showing that the thermodynamical limits of the entropy density in the microcanonical ensemble is the Legendre transform of the limit of the free energy density). One can also show it on the level of states, as we do here, both for classical \cite{Geo95} and, only very recently, for quantum systems \cite{MAMW13}. However the price of considering the thermodynamical limit---instead of the physically relevant regime of very large but finite sizes---is that no finite bounds can be obtained on the size of the regions on which the canonical and microcanonical states are close. 

In this respect Theorem \ref{equivalenceensembles} goes beyond the earlier work in several aspects:

\begin{itemize}

\item  It covers the general case of non translation-invariant models.

\item It is based on the assumption of a finite correlation length, which is simpler and more physical than the assumption of a unique phase region employed in \cite{Lima1, Lima2, MAMW13}. 

\item  It gives explicit finite size bounds; for quite big regions of order $O(N^{1/(d+1)} )$ the two ensembles already look the same.

\item It shows that the equivalence holds true even for microcanonical states with very small energy spread, of order $O(\log^{2d}(N))$ and substantially smaller than the value $O(\sqrt{N})$ that could have been expected. 

\item It covers more general states than the microcanonical, showing that the important conditions are that the state has small free energy or that it is concentrated around a fixed energy and has sufficiently large entropy.

\item It shows that that any two microcanonical states $\tau_{e, \delta}$ and $\tau_{e', \delta'}$ are locally equivalent whenever $|e - e'|N \leq O(\sqrt{N})$ and $O(\log^{2d}(N)/\sqrt{N}) \leq \delta + \delta' \leq O(1)$ (assuming $\rho_{T}$ has a finite correlation length).

\end{itemize}

It is an interesting open question to determine how small $\delta$ can be taken. We note that the eigenstate thermalization hypothesis (ETH) states that even for $\delta = 0$, i.e. for a single eigenstate, one should already have the same local expectations values as the corresponding microcanonical state \cite{Sre94}. However, while believed to hold true for several systems, there are known counterexamples to ETH, e.g., systems with many-body localization. 

\section{Proof Outline}

Our proof can be seen as a finite-size version of previous results \cite{Lima1, Lima2, MAMW13}  relating the (micro)canonical ensembles in the thermodynamical limit, in particular the recent work of M\"uller, Adlam, Masanes, and Wiebe \cite{MAMW13}, who showed the equivalence on the level of states for quantum systems in the thermodynamical limit. There the authors obtained the result from two observations, which we now briefly explain. 

Given a sequence of of translation-invariant Hamiltonians $H_{\Lambda_n}$ acting on finite volumes $\Lambda_n$ (with $n$ spins) with a well-defined thermodynamical limit, we define the (Hermholtz) free energy density as
\begin{equation} \label{freeenrtysenditydef}
f(T) := \inf \{ f_T(\omega) : \omega \hspace{0.2 cm} \text{translation-invariant state} \},
\end{equation}
with $f_T(\omega) = u(\omega) - T s(\omega)$ the free energy density for a translation-invariant state $\omega$ in the infinite lattice limit. In \cite{MAMW13} one is interested in the so-called \textit{one phase region}, in which there is only one state $\omega$ achieving the minimum in Eq.~(\ref{freeenrtysenditydef}), given by the KMS state associated to the sequence of finite volume Gibbs states. This uniqueness condition holds true if the finite volume Gibbs states have a finite correlation length. The first observation of the authors of \cite{MAMW13} is that 
\begin{equation} \label{equalitylimits}
\lim_{n \rightarrow \infty} \frac{1}{n} F_{T} \left(\tau_{u(T), o(\sqrt{N})}^{\Lambda_n} \right) =  \lim_{n \rightarrow \infty} \frac{1}{n} F_{T} \left(\rho_{T}^{\Lambda_n} \right),
\end{equation}
i.e., the free-energy density of the microcanonical ensemble converges to the free-energy density of the canonical ensemble (this fact is attributed to \cite{Sim93}). The second observation is that because of the uniqueness assumption in Eq.~(\ref{freeenrtysenditydef}), $ \tau_{u(T), o(\sqrt{N})}^{\Lambda_n}$ and $ \rho_{T}^{\Lambda_n} $ converge to the same state and therefore for any fixed region $\Lambda$:
\begin{equation}
\lim_{n \rightarrow \infty} \left \Vert   \tr_{\Lambda_n \backslash \Lambda}(\tau_{u(T), o(\sqrt{N})}^{\Lambda_n}) - \tr_{\Lambda_n \backslash \Lambda}(\rho_{T}^{\Lambda_n} ) \right \Vert_1 = 0.
\end{equation}

The proof of Theorem \ref{equivalenceensembles} will have a similar structure to the argument above. Indeed Proposition \ref{weakerProp} shows that every state of small enough free energy has its local reduced density matrices equal to the ones of the canonical state (if the latter has a finite correlation length). This is a finite-size analogue of the uniqueness of the minimizer in Eq.~(\ref{freeenrtysenditydef}) (which, as we mentioned, can also be derived from the assumption of a finite correlation length). To prove Theorem  \ref{equivalenceensembles} we show in Lemma \ref{opinqualitymicromacro} that 
\begin{equation} \label{aux343}
S(\tau_{u(T), \delta} \| \rho_{T}) \leq O(\log^{2d}(N)). 
\end{equation}
This is a finite-size analogue of Eq.~(\ref{equalitylimits}) and follows from a version of the Berry--Esseen Theorem for quantum lattice systems that we prove in Ref \cite{CBG15} (see Lemma \ref{thmBerryEsseenThm}). 
Note that since the maximum value of $S(\tau \| \rho_{T})$ is $O(n)$ (as the reference state is a thermal state), the Eq.~(\ref{aux343}) already suggests that the two states are not very different. Theorem \ref{equivalenceensembles} then follows directly from Proposition \ref{weakerProp} and Eq.~(\ref{aux343}).

We now give a quick summary of the argument behind the proof of Proposition \ref{weakerProp}. The proof in earnest is given in Section \ref{proof:prop1}. We use four variants of the quantum relative entropy. The first is the quantum Kullback--Leibler divergence defined before by:
\begin{equation}
S(\tau \| \rho) = \tr(\tau (\log \tau - \log \rho)).
\end{equation}
We also use its smoothed version:
\begin{equation}
S^{\varepsilon}(\tau \| \rho) := \min_{\tilde \tau \in B_{\varepsilon}(\tau)} S(\tilde \tau \| \rho),
\end{equation}
with $B_{\varepsilon}(\tau) := \{ \tilde{\tau} : \Vert \tau - \tilde \tau \Vert_1 \leq \varepsilon  \}$ the set of states that are $\varepsilon$-close to $\tau$. We also consider the max-relative entropy of two states $\tau$ and $\rho$ \cite{Datta09}:
\begin{equation}
S_{\max}(\tau \| \rho ) := \{ \min \lambda : \tau \leq 2^{\lambda} \rho   \},
\end{equation}
and its smooth version
\begin{equation}
S^{\varepsilon}_{\max}(\tau \Vert \rho ) := \min_{\tilde{\tau} \in B_{\varepsilon}(\tau)} S_{\max}(\tilde\tau \Vert \rho ).
\end{equation}
The relative entropies are related as follows
\begin{equation} \label{boundrelent}
S^{2\sqrt{\varepsilon}}(\tau \Vert \rho) \leq S^{2\sqrt{\varepsilon}}_{\text{max}}(\tau \Vert \rho) \leq  \frac{S(\tau \Vert \rho)+1}{\varepsilon} + \log \left( \frac{1}{1 - \varepsilon}  \right).
\end{equation}
The second inequality is known as quantum substate theorem \cite{substate,substate_improved}. 

For simplicity in this proof sketch we consider the one-dimensional case $d=1$, $N=n$, leaving the general case to the actual proof. 
The set of all intervals of length $l$ is given by $\mathcal{C}_l=\{C_1,\dots,C_{n-l+1}\}$ with $C_i=\{i,i+1,\dots,i+l-1\}$. Thus
\begin{equation}
\mathop{\mathbb{E}}_{C \in {\cal C}_l}\|\tau_{C}-\rho_{C}\|_1=\frac{1}{n-l+1}\sum_{i=1}^{n-l+1}\|\tau_{C_i}-\rho_{C_i}\|_1.
\end{equation}
We may now group the sets $C_i$ such that the $C_i$ within each group are separated from each other by a distance of $r$. There are at most $l+r$ such groups and within each group are $m\sim \frac{n}{l+r}$ sets $C_i$. Let us now focus on one group and let $C_1,\dots,C_m$ the sets in this group. Eqs. (\ref{simple:condition}), (\ref{boundrelent}), and the monotonicity under partial trace give
\begin{equation} \label{boundsmax2}
S_{\max}^{2\sqrt{\varepsilon}}(\tau_{C_1 \ldots C_m} \Vert \rho_{C_1 \ldots C_m}) \lesssim \sqrt{\epsilon n}.
\end{equation}
Lemma~\ref{tracenormboundfromxi} shows that if $\rho$ has finite correlation length $\xi$ then 
\begin{equation}
\Vert \rho_{C_1 \ldots C_m} - \rho_{C_1} \otimes \cdots \otimes \rho_{C_m} \Vert_1 \lesssim m n^zD^{2l-r/\xi},
\end{equation}
which can be made arbitrarily small by increasing $r$.
By the data processing inequality  for the smooth max-relative entropy \cite{Datta09} (see Lemma~\ref{datarenner} in Section~\ref{auxLemmas})  and Eq.~(\ref{boundsmax2})
one then has
\begin{equation}
S_{\max}^{\varepsilon^\prime}(\tau_{C_1 \ldots C_m} \Vert \rho_{C_1} \otimes \cdots \otimes \rho_{C_m})\lesssim \sqrt{\epsilon n}.
\end{equation}
with $\varepsilon^\prime=2\sqrt{\varepsilon}+\sqrt{8m n^zD^{2l-r/\xi}}$.
From Eq.~(\ref{boundrelent}), we find that the above bound also holds for $S^{\varepsilon^\prime}(\tau_{C_1 \ldots C_m} \Vert \rho_{C_1} \otimes \cdots \otimes \rho_{C_m})$. Then by subadditivity of the von Neumann entropy  it follows that 
\begin{equation}
\sum_{i=1}^mS^{\varepsilon^\prime}(\tau_{C_i} \Vert \rho_{C_i})\lesssim \sqrt{\epsilon n}.
\end{equation}
Pinsker's inequality gives
\begin{equation}
 \sqrt{S^{\varepsilon'}( \tau_{C_i} \| \rho_{C_i})} \gtrsim \Vert \tau_{C_i} -  \rho_{C_i} \Vert_1 - \varepsilon'
\end{equation}
and hence, due to our choice of $l\lesssim \sqrt{\epsilon n}$ given by Eq. (\ref{simple:condition}) and since by construction $m \sim n / (l+r)$, we have
\begin{equation}
\begin{split}
\mathop{\mathbb{E}}_{C \in {\cal C}_l}\|\tau_{C}-\rho_{C}\|_1&\lesssim \varepsilon'+\frac{\sqrt{m}(l+r)}{n-l+1}(\epsilon n)^{1/4}
\lesssim \sqrt{\varepsilon}+\sqrt{n^{z+1} D^{2\sqrt{\epsilon n}-r/\xi}}+\sqrt{\epsilon+r\sqrt{\epsilon /n}},
\end{split}
\end{equation}
and the result then follows from setting $r\sim \sqrt{\epsilon n}$.

\vspace{0.3 cm}
\section{Proofs}

\subsection{Proof of Proposition~\ref{weakerProp}}
\label{proof:prop1}

Let $1\le l\le n$ and ${\cal C}_l$ the set of all cubes in $\Lambda$ with edge length $l$, i.e., 
\begin{equation}
{\cal C}_l=\left\{C_i\,\big|\,i\in\{1,n-l+1\}^d\right\},\;\;\; C_i=i+\{0,\dots,l-1\}^d.
\end{equation}
Writing $\Lambda_l=\{1,n-l+1\}^d$, we thus have
\begin{equation}
\begin{split}
\mathop{\mathbb{E}}_{C \in {\cal C}_l}\|\tau_{C}-\rho_{C}\|_1 &=\frac{1}{(n-l+1)^d}\sum_{i\in\Lambda_l}\|\tau_{C_i}-\rho_{C_i}\|_1.
\end{split}
\end{equation}
\begin{figure}[t]
\begin{center}
\includegraphics[width=0.9\textwidth]{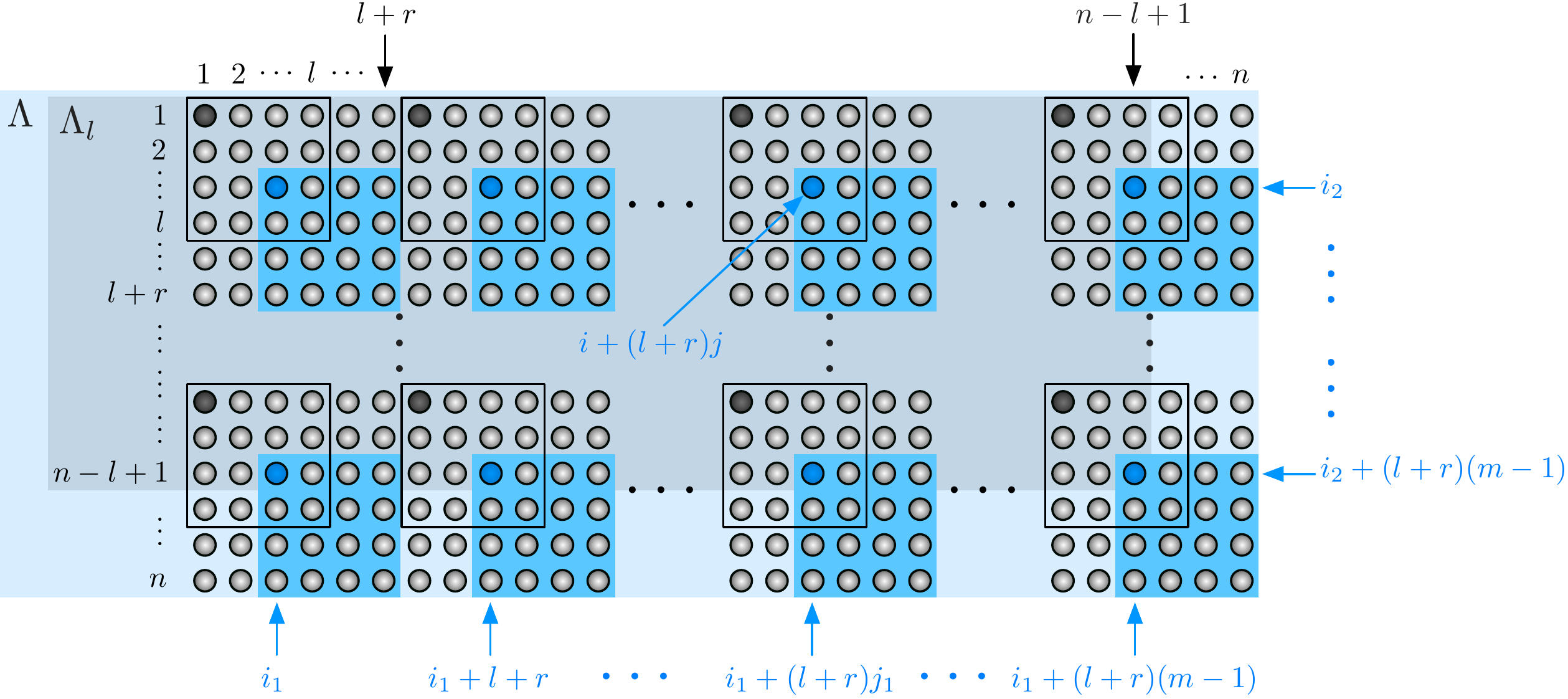}
\end{center}
\caption{\label{decomp} The lattice $\Lambda=\{1,\dots,n\}^d$ for $d=2$. The set  of all cubic subsets of edge length $l$ is $\mathcal{C}_l=\{C_i\subset\Lambda\,|\,i\in\Lambda_l\}$, where $C_i=i+\{0,\dots,l-1\}^d$ and $\Lambda_l=\{1,n-l+1\}^d$, which we decompose as $\Lambda_l=\bigcup_{i\in\{1,\dots,l+r\}^d}{\cal L}^{(r,l)}_{i}$.
Blue sites are the elements of ${\cal L}^{(r,l)}_{i}$ for $l=4$, $r=2$, and $i=(i_1,i_2)=(2,3)$. Blue squares are the corresponding cubic subsets of edge length $l$, $C_{i+(l+r)j}$, which are separated by $r$ sites. Black sites are the elements of ${\cal L}^{(r,l)}_{i}$ for $i=(i_1,i_2)=(1,1)$ and black squares indicate the corresponding $C_{i+(l+r)j}$. Here, $j\in\{0,\dots,m-1\}^d$ and $m=\lceil\frac{n-l+1}{l+r}\rceil$.}
\end{figure}%
We now decompose $\Lambda_l$ into a union of sets  ${\cal L}^{(r,l)}_{i}$ such that for every $j,j^\prime\in{\cal L}^{(r,l)}_{i}$, $j\ne j^\prime$, we have $\text{dist}(C_j,C_{j^\prime})> r$ as in Fig.~\ref{decomp}: Define $m:=\lceil \frac{n-l+1}{l+r}\rceil$ and
\begin{equation}
\begin{split}
{\cal L}^{(r,l)}_{i}&=
\bigcup_{j\in\{0,\dots,m-1\}^d}\{i+(l+r)j\}\cap\Lambda_l\\
&=\left\{i+(l+r)j\in\Lambda_l\,\big|\,j\in\{0,\dots,m-1\}^d\right\}
\end{split}
\end{equation}
such that
\begin{equation}
\label{how_many}
\begin{split}
\Lambda_l&=\bigcup_{i\in\{1,\dots,l+r\}^d}{\cal L}^{(r,l)}_{i},\;\;\;|\Lambda_l|=\sum_{i\in\{1,\dots,l+r\}^d}|{\cal L}^{(r,l)}_{i}|,
\end{split}
\end{equation}
and
\begin{equation}
\label{decomp_start}
\begin{split}
\mathop{\mathbb{E}}_{C \in {\cal C}_l}\|\tau_{C}-\rho_{C}\|_1 &=\frac{1}{|\Lambda_l|}\sum_{i\in\{1,\dots,l+r\}^d}\sum_{j\in\mathcal{L}^{(r,l)}_i}\|\tau_{C_j}-\rho_{C_j}\|_1.
\end{split}
\end{equation}
We now need the following lemma.

\begin{lemma} \label{starting_lemma} Let $0<\epsilon< 1$ and $C_1,\dots,C_M\subset \Lambda$ with $\text{dist}(C_i,C_j)>0$ for $i\ne j$.  Let $\tau,\rho$ states such that $\rho$ has $(\xi,z)$-exponentially decaying correlations and let
\begin{equation}
\kappa:=2^{\frac{S(\tau\|\rho)+1}{\epsilon}+\log\left(\frac{1}{1-\epsilon}\right)}
\sum_{j=2}^MD^{2|C_j|}N^z\me^{-\text{dist}(C_1\cup\cdots\cup C_{j-1},C_j)/\xi}<1.
\end{equation}
Then there is a state $\pi_{C_1\cdots C_M}$ such that 
\begin{equation}
S_{\max}(\pi_{C_1\cdots C_M} \| \rho_{C_1}\otimes\cdots\otimes\rho_{C_m}) \leq \frac{S(\tau\|\rho)+1}{\epsilon}+\log\left(\frac{1}{1-\epsilon}\right)+\log\left(\frac{1}{1-\kappa}\right)
\end{equation}
and $\|\pi_{C_1\cdots C_M}-\tau_{C_1\cdots C_M}\|_1\le 2\sqrt{\epsilon}+\sqrt{8\kappa}$.
\end{lemma}

\begin{proof} 
The quantum substate theorem \cite{substate}
in the version of Ref.~\cite{substate_improved} implies the existence of a state $\tilde{\pi}_{C_1\cdots C_M}$ such that 
\begin{equation}
\label{having_pi}
\|\tilde{\pi}_{C_1\cdots C_M}-\tau_{C_1\cdots C_M}\|_1\le 2\sqrt{\epsilon}
\end{equation}
and
 \begin{equation}
 S_{\text{max}}(\tilde{\pi}_{C_1\cdots C_M}\|\rho_{C_1\cdots C_M})\le \frac{S(\tau_{C_1\cdots C_M}\|\rho_{C_1\cdots C_M})+1}{\epsilon}+\log\left(\frac{1}{1-\epsilon}\right),
 \end{equation} 
 such that, using the monotonicity under partial trace of the quantum relative entropy,
  \begin{equation}
 S_{\text{max}}(\tilde{\pi}_{C_1\cdots C_M}\|\rho_{C_1\cdots C_M})\le \frac{S(\tau\|\rho)+1}{\epsilon}+\log\left(\frac{1}{1-\epsilon}\right)=:\lambda.
 \end{equation}  
  This proves the statement for $M=1$. For $M>1$, we use Lemma~\ref{tracenormboundfromxi} to find
\begin{equation}
\|\rho_{C_1}\otimes\cdots\otimes\rho_{C_M}-\rho_{C_1\cdots C_M}\|_1\le \sum_{j=2}^MD^{2|C_j|}\text{cor}_{\rho}(C_1\cdots C_{j-1},C_j)=:c.
\end{equation}
By Lemma~\ref{datarenner}, if $\kappa:=2^\lambda c<1$ then there is a state $\pi_{C_1\cdots C_M}$ such that
 \begin{equation}
\label{eating_it_too}
S_{\max}(\pi_{C_1\cdots C_M} \| \rho_{C_1}\otimes\cdots\otimes\rho_{C_M}) \leq \lambda+\log\left(\frac{1}{1-\kappa}\right)
\end{equation}
and $\|\pi_{C_1\cdots C_M}-\tilde{\pi}_{C_1\cdots C_M}\|_1\le \sqrt{8\kappa}$, which, in combination with Eqs.~(\ref{having_pi}) and (\ref{eating_it_too}), proves the statement for $M>1$ as by the triangle inequality $\|\pi_{C_1\cdots C_M}-\tau_{C_1\cdots C_M}\|_1\le 2\sqrt{\epsilon}+\sqrt{8\kappa} $.
\end{proof}

We are now in the position to prove the following stronger version of Proposition~\ref{weakerProp}.

\begin{prop} \label{strongerProp} Let $\epsilon>0$, the states $\rho$, $\tau$, and $l\in\nn$, $1\le l\le \frac{n+1}{2}$, such that $\rho$ has $(\xi,z)$-exponentially decaying correlations and such that 
\begin{equation}
\label{thm:condition}
\left\lceil W\!\left((2^d-1)^{1/d}\frac{n-l+1}{ \epsilon^{1/d}\xi d}
2^{ \frac{S(\tau\|\rho)+3/2}{\epsilon d}} D^{2l^d/d}n^z\me^{\frac{l-1}{\xi d}}\right)\xi d\right\rceil^d
\frac{S(\tau\|\rho)+2}{\epsilon}\le \epsilon (n-l+1)^d.
\end{equation}
Then
\begin{equation}
\label{thm:result}
\mathop{\mathbb{E}}_{C \in {\cal C}_l}\|\tau_{C}-\rho_{C}\|_1 \le \left(\sqrt{2}+2+\sqrt{\ln(2)}\right)\sqrt{2\epsilon}.
\end{equation}
Here, $\lceil\cdot\rceil$ denotes the smallest integer not less than $\cdot$ and $W$ the solution to $z=W(z)\me^{W(z)}$, $W(z)\ge 0$ (one of the real branches of the Lambert $W$ function).
\end{prop}

\begin{proof}
We set out to combine Lemma~\ref{starting_lemma} with the following basic properties of the quantum relative entropy and trace norm.
\begin{itemize}
\item[{[a]}] (Pinsker's inequality) $\Vert \rho - \sigma \Vert_1^2\leq  \ln(4)S(\rho \| \sigma)  $,

\item[{[b]}]  (Relation with $S_{\max}$ \cite{Datta09}) $S(\rho \| \sigma) \leq S_{\max}(\rho \| \sigma)$,

\item[{[c]}]  (Super-additivity\footnote{
This is an easy consequence of subadditivity of entropy. Indeed,
$S(\pi_{A_1 \cdots A_M} \| \rho_{A_1} \otimes \cdots \otimes \rho_{A_M}) = - S(\pi_{A_1 \cdots A_M}) - \tr( \pi_{A_1 \cdots A_M} \log(\rho_{A_1} \otimes \cdots \otimes \rho_{A_M}))  \geq - \sum_{j=1}^M S(\pi_{A_j}) - \tr( \pi_{A_1 \cdots A_M} \log(\rho_{A_1} \otimes \cdots \otimes \rho_{A_M})) = \sum_{j=1}^M S(\pi_{A_j} \| \rho_{A_j})$.
}) $\sum_{j=1}^M S(\pi_{A_j} \| \rho_{A_j})\leq S(\pi_{A_1 \cdots A_M} \| \rho_{A_1} \otimes \cdots \otimes \rho_{A_M})$,
\item[{[d]}]  (Monotonicity under partial trace) $\|\rho_A-\sigma_A\|_1\le \|\rho-\sigma\|_1$ for all $A\subset\Lambda$.
\end{itemize}
To this end, let $0<\epsilon\le 1/2$ and 
\begin{equation}
\label{kappa}
\kappa:=2^{ \frac{S(\tau\|\rho)+3/2}{\epsilon}}(m^d-1) D^{2l^d}N^z\me^{-(r+1)/\xi}\le \epsilon.
\end{equation}
Then, by Lemma~\ref{starting_lemma}, we have that for each $i\in\Lambda_l$ there is a state $\pi_{C_1\cdots C_{M_i}}$, $M_i=|\mathcal{L}^{(\delta,l)}_i|$, such that (we use   that $\log(\frac{1}{1-\epsilon})\le \frac{1}{2\epsilon}$ for $0<\epsilon\le 1/2$)
\begin{equation}
\label{Smax}
\begin{split}
S_{\max}(\pi_{C_1\cdots C_{M_i}} \| \rho_{C_1}\otimes\cdots\otimes\rho_{C_{M_i}}) \le \frac{S(\tau\|\rho)+2}{\epsilon}
\end{split}
\end{equation}
 and
 \begin{equation}
 \label{1-norm}
\|\pi_{C_1\cdots C_{M_i}}-\tau_{C_1\cdots C_{M_i}}\|_1\le 2\sqrt{\epsilon}+\sqrt{8\epsilon}.
\end{equation}
Then, starting with Eq.~\eqref{decomp_start},
\begin{equation}
\label{main_thing}
\begin{split}
\mathop{\mathbb{E}}_{C \in {\cal C}_l}\|\tau_{C}-\rho_{C}\|_1 &\le\frac{1}{|\Lambda_l|}\sum_{i\in\{1,\dots,l+r\}^d}\sum_{j\in\mathcal{L}^{(r,l)}_i}\|\tau_{C_j}-\pi_{C_j}\|_1\\
&\hspace{2cm}+\frac{1}{|\Lambda_l|}\sum_{i\in\{1,\dots,l+r\}^d}|\mathcal{L}^{(r,l)}_i|^{1/2}\sqrt{\sum_{j\in\mathcal{L}^{(r,l)}_i}\|\pi_{C_j}-\rho_{C_j}\|^2_1}\\
&\underset{\text{a,d}}{\le}\frac{1}{|\Lambda_l|}\sum_{i\in\{1,\dots,l+r\}^d}\sum_{j\in\mathcal{L}^{(r,l)}_i}\|\tau_{C_1\cdots C_{M_i}}-\pi_{C_1\cdots C_{M_i}}\|_1\\
&\hspace{2cm}+\frac{1}{|\Lambda_l|}\sum_{i\in\{1,\dots,l+r\}^d}|\mathcal{L}^{(r,l)}_i|^{1/2}\sqrt{\ln(4)\sum_{j\in\mathcal{L}^{(r,l)}_i}S(\pi_{C_j}\|\rho_{C_j})}\\
&\underset{\text{b,c,\eqref{how_many},\eqref{1-norm}}}{\le}(2+\sqrt{8})\sqrt{\epsilon}\\
&\hspace{1cm}
+\frac{1}{|\Lambda_l|}\sum_{i\in\{1,\dots,l+r\}^d}|\mathcal{L}^{(r,l)}_i|^{1/2}\sqrt{\ln(4)S_{\text{max}}(\pi_{C_1\cdots C_{M_i}}\|\rho_1\otimes\cdots\otimes\rho_{M_i})}\\
&\underset{\text{\eqref{Smax}}}{\le}(2+\sqrt{8})\sqrt{\epsilon}
+\frac{1}{|\Lambda_l|}\sum_{i\in\{1,\dots,l+r\}^d}|\mathcal{L}^{(r,l)}_i|^{1/2}\sqrt{\ln(4)\frac{S(\tau\|\rho)+2}{\epsilon}}\\
&\underset{\text{\eqref{how_many}}}{\le} (2+\sqrt{8})\sqrt{\epsilon}
+\sqrt{\ln(4)\frac{(l+r)^{d}}{(n-l+1)^d}\frac{S(\tau\|\rho)+2}{\epsilon}},
\end{split}
\end{equation}
where we used the (triangle and) Cauchy--Schwarz inequality to obtain the last (first) line. Now, for $m=1$ we have $\kappa=0$ and for $m\ge 2$ we have
$\frac{n-l+1}{l+r}\ge 1$ such that $\kappa\le \epsilon$  is implied by (see Eq.~\eqref{kappa})
\begin{equation}
\label{kappa2}
2^{ \frac{S(\tau\|\rho)+3/2}{\epsilon d}}\frac{n-l+1}{\epsilon^{1/d}\xi d}(2^d-1)^{1/d} D^{2l^d/d}n^z\me^{\frac{l-1}{\xi d}}\le \frac{l+r}{\xi d}\me^{\frac{l+r}{\xi d}},
\end{equation}
which also ensures that $r\ge 0$ for integer $r$.\footnote{The left hand side of Eq.~\eqref{kappa2} is lower bounded by $(n-l+1) \me^{\frac{l-1}{\xi d}}/(\xi d)$ such that for $n+1\ge 2l$ we have $(l+r) \me^{\frac{l+r}{\xi d}}/(\xi d)\ge l \me^{\frac{l-1}{\xi d}}/(\xi d)>(l-1) \me^{\frac{l-1}{\xi d}}/(\xi d)$.}
Hence, setting 
\begin{equation}
r=-l+\left\lceil\xi d W\left(
2^{ \frac{S(\tau\|\rho)+3/2}{\epsilon d}}\frac{n-l+1}{ \epsilon^{1/d}\xi d}(2^d-1)^{1/d} D^{2\frac{l^d}{d}}n^z\me^{\frac{l-1}{\xi d}}\right)\right\rceil,
\end{equation}
we have that Eq.~\eqref{thm:condition} implies Eq.~\eqref{thm:result}, which trivially also holds for $\epsilon>1/2$ as $\|\tau_C-\rho_C\|_1\le 2$ for all states $\tau$, $\rho$.
\end{proof}

Finally, using the bound $W(z)\le \ln(z+1)$ we find
that  Eq.~\eqref{simple:condition} implies Eq.~\eqref{thm:condition}, which proves Proposition~\ref{weakerProp}.

\subsection{Proof of Theorem \ref{equivalenceensembles}}  \label{proofmain}
We proof the following stronger version of Theorem~\ref{equivalenceensembles}.

\begin{thm} \label{equivalenceensembles_stronger} Let the canonical state $\rho_T$ (corresponding to a $k$-local Hamiltonian as in Eq.~\eqref{ham}) with energy density $u(T)$ and specific heat capacity $c(T)$ have $(\xi,z)$-exponentially decaying correlations. Let
the microcanonical state $\tau_{e,\delta}$ have mean energy such that
\begin{equation}
 |e-u(T)|\le  \sqrt{c(T)T^2/N}
 \end{equation}
  and energy spread such that
\begin{equation}
28\Delta_{k,\xi,z,T}\sqrt{c(T)T^2}\frac{\ln^{2d}(N)}{\sqrt{N}}\le \delta\le \sqrt{c(T)T^2}.
\end{equation}
Let $\epsilon>0$ and write
\begin{equation}
s=\frac{1}{\epsilon}\log\left(\frac{\sqrt{N}}{\Delta_{k,\xi,z,T}\ln^{2d}(N)}\me^{56\sqrt{c(T)}\Delta_{k,\xi,z,T}\ln^{2d}(N)}\right)+\frac{2}{\epsilon}.
\end{equation}
If $l\in\nn$, $1\le l\le \frac{n+1}{2}$, and $s$ are such that
\begin{equation}
\left\lceil W\!\left((2^d-1)^{1/d}\frac{n-l+1}{ \epsilon^{1/d}\xi d}
2^{ s/d} D^{2l^d/d}n^z\me^{\frac{l-1}{\xi d}}\right)\xi d\right\rceil^d
s\le \epsilon (n-l+1)^d
\end{equation}
then
\begin{equation}
\mathop{\mathbb{E}}_{C \in {\cal C}_l}\|\tau_{C}-\rho_{C}\|_1 \le \left(\sqrt{2}+2+\sqrt{\ln(2)}\right)\sqrt{2\epsilon}.
\end{equation}
Here, $\lceil\cdot\rceil$ denotes the smallest integer not less than $\cdot$ and $W$ the solution to $z=W(z)\me^{W(z)}$, $W(z)\ge 0$ (one of the real branches of the Lambert $W$ function).
\end{thm}

This theorem is a direct consequence of Proposition~\ref{strongerProp} and the following lemma. Theorem~\ref{equivalenceensembles}  follows from
Proposition~\ref{weakerProp}, the following lemma,  and the bound (we use the assumption $N>2$, the fact that $\Delta_{k,\xi,z,T}\ge (c(T)T^2)^{-3/2}=N^{3/2}(\tr(H^2\rho_T)-(\tr[H\rho_T])^2)^{-3/2}\ge N^{-3/2}$, and $\epsilon\le 1/2$, which, as we recall, is w.l.o.g.)
\begin{equation}
\label{simplify_bound}
\begin{split}
& \frac{ \log\left(
\frac{\sqrt{N}}{\Delta_{k,\xi,z,T}\ln^{2d}(N)}
\me^{56\sqrt{c(T)}\Delta_{k,\xi,z,T}\ln^{2d}(N)}\right)
+3}{\epsilon}
+\ln(N^{z+1}) \\
&\hspace{6cm} \le\frac{ 
56\sqrt{c(T)}\Delta_{k,\xi,z,T}\ln^{2d}(N)
+(5+\epsilon z)\ln(N)}{\epsilon \ln(2)}.
\end{split}
\end{equation}

\begin{lemma}  \label{opinqualitymicromacro} Let the canonical state $\rho_T$ (corresponding to a $k$-local Hamiltonian as in Eq.~\eqref{ham}) with energy density $u(T)$ and specific heat capacity $c(T)$ have $(\xi,z)$-exponentially decaying correlations. Let
the state $\tau\in \mathcal{D} ( \text{span}[\{|\nu\rangle\}_{\nu\in M_{e,\delta}}])$ with
\begin{equation}
 |e-u(T)|\le  \sqrt{c(T)T^2/N}
 \end{equation}
  and
\begin{equation}
28\Delta_{k,\xi,z,T}\sqrt{c(T)T^2}\frac{\ln^{2d}(N)}{\sqrt{N}}\le \delta\le \sqrt{c(T)T^2}.
\end{equation}
Then
\begin{equation}
\begin{split}
S(\tau\|\rho_T)&\le-S(\tau)+\log(|M_{e, \delta}|)+\log\left(
\frac{\sqrt{N}}{\Delta_{k,\xi,z,T}\ln^{2d}(N)}
\me^{56\sqrt{c(T)}\Delta_{k,\xi,z,T}\ln^{2d}(N)}\right).
\end{split}
\end{equation}
\end{lemma}

\begin{proof}
We write $\sigma^2=NT^2c(T)$, $\mu=Nu(T)$, and define
\begin{equation}
Z(T, e , \delta) := \sum_{\nu \in  M_{e, \delta}} e^{- E_\nu / T}.
\end{equation}
Then, for any $\delta,\tilde \delta >0$, any $e,\tilde e\in\rr$, and any state $\tau=\sum_{\nu,\nu^\prime\in M_{e,\delta}}\tau_{\nu,\nu^\prime}|\nu\rangle\langle \nu^\prime|$
\begin{equation}
\label{relEntStart}
\begin{split}
S(\tau\|\rho_T)&=-S(\tau)-\text{tr}[\tau\log(\rho_T)]\\
&=-S(\tau)+\sum_{\nu\in M_{e,\delta}}\tau_{\nu,\nu}\log\left(\frac{Z(T)}{Z(T,\tilde e,\tilde \delta)}Z(T,\tilde e,\tilde \delta)\me^{E_\nu/T}\right)\\
&\le-S(\tau)+\log\left(\frac{Z(T)}{Z(T,\tilde e,\tilde \delta)}|M_{\tilde e,\tilde \delta}|\me^{(eN-\tilde eN+\delta\sqrt{N}+\tilde\delta\sqrt{N})/T}\right).
\end{split}
\end{equation}
We will choose $\tilde \delta$ and $\tilde e$ below after Eq.~\eqref{back_to_relEnt} and continue with bounding
\begin{equation}
\begin{split}
\frac{Z(T, \tilde e,\tilde \delta) }{Z(T)}&= \sum_{\nu \in M_{\tilde e,\tilde \delta}} \langle \nu | \rho_{T}|\nu\rangle
=\sum_{\nu:\,|E_\nu-\tilde eN|\le \delta\sqrt{N}}  \langle \nu |  \rho_{T} | \nu \rangle\\
&=\sum_{\nu:\, E_\nu\le \tilde eN+\tilde \delta\sqrt{N}}  \langle \nu |  \rho_{T} | \nu \rangle-\sum_{\nu:\,E_\nu<\tilde eN-\delta\sqrt{N}}  \langle \nu |  \rho_{T} | \nu \rangle\\
&\ge F(\tilde eN+\tilde \delta\sqrt{N})-F(\tilde eN-\tilde \delta\sqrt{N})\\
&\ge G(\tilde eN+\tilde \delta\sqrt{N})-G(\tilde eN-\tilde \delta\sqrt{N})-2\sup_x|F(x)-G(x)|,
\end{split}
\end{equation}
where the (Gaussian) cumulative distribution ($G$) $F$ is defined in Lemma~\ref{thmBerryEsseenThm}.
By the mean value theorem, for some $x\in (-\tilde \delta\sqrt{N},\tilde \delta\sqrt{N})$,
\begin{equation}
\begin{split}
G(\tilde eN+\tilde \delta\sqrt{N})-G(\tilde eN-\tilde \delta\sqrt{N})&=2\tilde \delta\sqrt{N}\frac{1}{\sqrt{2\pi\sigma^2}}\me^{-\frac{(\tilde eN-\mu+x)^2}{2\sigma^2}},
\end{split}
\end{equation}
i.e., by Lemma~\ref{thmBerryEsseenThm} 
\begin{equation}
\begin{split}
\frac{Z(T,\tilde e,\tilde \delta) }{Z(T)}&\ge \frac{2}{\sigma}\tilde  \delta\sqrt{N}\frac{1}{\sqrt{2\pi}}\me^{-\frac{(|\tilde eN-\mu|+\tilde \delta\sqrt{N})^2}{2\sigma^2}}-2\Delta_{k,\xi,z,T}\frac{\ln^{2d}(N)}{\sqrt{N}}.
\end{split}
\end{equation}
Hence, for
\begin{equation}
\label{condition_on_delta}
\begin{split}
\frac{2}{\sigma}\tilde  \delta\sqrt{N}\frac{1}{\sqrt{2\pi}}\me^{-\frac{(|\tilde eN-\mu|+\tilde \delta\sqrt{N})^2}{2\sigma^2}}\ge 3\Delta_{k,\xi,z,T}\frac{\ln^{2d}(N)}{\sqrt{N}}
\end{split}
\end{equation}
we have
\begin{equation}
\label{back_to_relEnt}
\begin{split}
\frac{Z(T,\tilde e,\tilde \delta) }{Z(T)}&\ge \Delta_{k,\xi,z,T}\frac{\ln^{2d}(N)}{\sqrt{N}}.
\end{split}
\end{equation}
We now  set $\tilde eN=eN+\delta\sqrt{N}-\tilde \delta\sqrt{N}$. 
Assuming $|eN-\mu|\le\sigma$ and $\tilde \delta\le \delta\le \sigma/\sqrt{N}$, 
this choice implies $|\tilde eN-\mu|+\tilde \delta\sqrt{N}\le |eN-\mu|+\delta\sqrt{N}\le 2\sigma$, i.e., the condition in Eq.~\eqref{condition_on_delta} is implied by
\begin{equation}
\begin{split}
\delta_0:= \frac{3\sqrt{2\pi}\Delta_{k,\xi,z,T}}{2}\me^{2}\frac{\ln^{2d}(N)}{\sqrt{N}}\frac{\sigma}{\sqrt{N}}\le \tilde  \delta\le \delta\le\frac{\sigma}{\sqrt{N}}.
\end{split}
\end{equation}
We return to Eq.~\eqref{relEntStart} and set $\tilde \delta=\delta_0$ to find
\begin{equation}
\begin{split}
S(\tau\|\rho_T)&\le-S(\tau)+\log\left(
\frac{\sqrt{N}}{\Delta_{k,\xi,z,T}\ln^{2d}(N)}
|M_{\tilde e,\tilde \delta}|\me^{2\delta_0\sqrt{N}/T}\right).
\end{split}
\end{equation}
Finally, as we assumed that $\delta\ge \delta_0$, $|E_k-\tilde eN|\le \sqrt{N}\delta_0$ implies $|E_k- eN|\le |E_k-\tilde eN|+|\tilde eN-eN|\le |\delta\sqrt{N}- \delta_0\sqrt{N}|+\delta_0\sqrt{N}=\delta\sqrt{N}$
such that $|M_{\tilde e,\delta_0}|\le |M_{e,\delta}|$.
\end{proof}

\subsection{Proof of Corollary \ref{corollaryProp}}
The first part follows directly from Proposition~\ref{weakerProp}, Lemma~\ref{opinqualitymicromacro}, and the bound in Eq.~\eqref{simplify_bound}.

By Ref.~\cite{PSW06}, for any $\varepsilon>0$, with probability at least
\begin{equation}
1-2\exp\left(-\frac{|M_{e,\delta}|\varepsilon^2}{18\pi^3}\right)
\end{equation}
one has
\begin{equation}
\|\pi_C-(\tau_{e,\delta})_C\|_1\le \varepsilon+\sqrt{\frac{d_S}{d_E^{eff}}}\le\varepsilon+\frac{D^{l^d}}{\sqrt{|M_{e,\delta}|}}.
\end{equation}
Further, as in the proof of Lemma~\ref{opinqualitymicromacro} (see the discussion around Eq.~\eqref{back_to_relEnt}),
\begin{equation}
\begin{split}
S(\varrho_T)&=Nu(T)/T+\ln\left(Z(T)\right)\\
&\le Nu(T)/T-eN/T+\ln\left(|M_{e,\delta}|\right)+\ln\left(\frac{\sqrt{N}}{\Delta_{k,\xi,z,T}\ln^{2d}(N)}\right)+\delta_0\sqrt{N}/T\\
&\le 2\sigma/T+\ln\left(|M_{e,\delta}|\right)+\ln\left(N^2\right)\le 2(\sqrt{c(T)}+1)\sqrt{N}+\ln\left(|M_{e,\delta}|\right),
\end{split}
\end{equation}
where we used $N>2$ and $\Delta_{k,\xi,z,T}\ge (c(T)T^2)^{-3/2}=N^{3/2}(\tr(H^2\rho_T)-(\tr[H\rho_T])^2)^{-3/2}\ge N^{-3/2}$ to obtain the last line. 
Hence, with probability at least
\begin{equation}
1-2\exp\left(-\frac{\varepsilon^2}{18\pi^3}\exp\left[N\left(s(T)-\frac{2(\sqrt{c(T)}+1)}{\sqrt{N}}\right)\right]\right)
\end{equation}
we have
\begin{equation}
\|\pi_C-(\tau_{e,\delta})_C\|_1\le \varepsilon+D^{l^d}\exp\left[-\frac{N}{2}\left(s(T)-\frac{2\sqrt{c(T)}+2}{\sqrt{N}}\right)\right].
\end{equation}

\subsection{Auxiliary Lemmas}
\label{auxLemmas}
The following auxiliary lemmas were used in the proofs above.
The first is the main result of \cite{CBG15}, a Berry--Esseen bound for quantum lattice systems.

\begin{lemma}  \label{thmBerryEsseenThm}
On $\Lambda = \{1,\dots,n\}^{\times d}$ with $N = n^d>1$ sites let $H$ be a $k$-local Hamiltonian as in Eq.~\eqref{ham}  and let $\rho$ a state with $(\xi,z)$-exponentially decaying correlations.  
Let
\begin{equation}
\label{defF}
F(x)=\sum_{k:\,   E_k \le x}\langle k|\rho |k\rangle,\;\;\;\mu=\tr(\rho H),\;\;\;\sigma^2=\tr\left( \rho (H-\mu)^2\right),
\end{equation}
and 
\begin{equation}
G(x)=\frac{1}{\sqrt{2\pi\sigma^2}}\int_{-\infty}^x\md y\,\me^{-\frac{(y-\mu)^2}{2\sigma^2}}
\end{equation}
the Gaussian cumulative distribution with mean $\mu$ and variance $\sigma^2$.
Then
\begin{equation}  \label{eqBEmainthm}
\sup_x|F(x)-G(x)|  \le \Delta \frac{\ln^{2d}(N)}{\sqrt{N}},
\end{equation}
where
\begin{equation}  
\label{BEdelta}
\Delta= C_d \frac{(\max\{k,\xi\}(z+1))^{2d}}{\sigma/\sqrt{N}}\max\left\{\frac{1}{\max\{k,\xi\}(z+1)\ln(N)},\frac{1}{\sigma^2/N}\right\}
\end{equation}
and $C_d\ge 1$ depends only on the dimension of the lattice.
\end{lemma}

 The next lemma was originally proven by Datta and Renner in \cite{DR09}, in a different formulation, and appeared in a form equivalent to the one bellow as Lemma C.5 of \cite{BP09}.

\begin{lemma} \label{datarenner}
Let $\tilde{\pi}, \rho, \tilde{\rho} \in {\cal D}(\mathbb{H})$ be such that $S_{\max}(\tilde{\pi} \| \rho) \leq \lambda$ and $\kappa:=2^\lambda\|\tilde{\rho} -\rho\|_1<1$. Then 
there is a state $\pi$ such that
\begin{equation}
S_{\max}(\pi \| \tilde{\rho}) \leq \lambda+\log\left(\frac{1}{1-\kappa}\right)
\end{equation}
and $\|\tilde{\pi}-\pi\|_1\le \sqrt{8\kappa}$.
\end{lemma}

\begin{proof}
The statement follows from Lemma C.5 of \cite{BP09} with $Y = 2^{\lambda} \tilde{\rho}$ and $\Delta = 2^{\lambda} |\rho - \tilde{\rho}|$. For completeness, we give the proof following \cite{BP09} and \cite{DR09}: 
We have 
\begin{equation}
\tilde{\pi}\le 2^\lambda \rho\le 2^\lambda \tilde{\rho}+2^\lambda |\rho-\tilde{\rho}| =Y+\Delta.
\end{equation}
Let $T=Y^{1/2}(Y+\Delta)^{-1/2}$ (with the inverse the generalized  Moore--Penrose pseudoinverse) and $\pi=T\tilde{\pi}T^\dagger/\text{tr}[T^\dagger T\tilde{\pi}]$.
We find $T\tilde{\pi}T^\dagger\le Y^{1/2}(Y+\Delta)^{-1/2}(Y+\Delta)(Y+\Delta)^{-1/2}Y^{1/2}\le Y$. Further, $T^\dagger T=(Y+\Delta)^{-1/2}Y(Y+\Delta)^{-1/2}\le \id$  such that
\begin{equation}
\begin{split}
 \text{tr}[(\id-T^\dagger T)\tilde{\pi}]&\le\text{tr}[(\id-T^\dagger T)(Y+\Delta)]
= \text{tr}[Y+\Delta]-\text{tr}[T^\dagger T(Y+\Delta)]\\
&= \text{tr}[Y+\Delta]-\text{tr}[(Y+\Delta)^{-1/2}(Y+\Delta)(Y+\Delta)^{-1/2}(Y+\Delta)]\\
&\hspace{1cm}+\text{tr}[(Y+\Delta)^{-1/2}\Delta(Y+\Delta)^{-1/2}(Y+\Delta)]\\
&=\text{tr}[\Delta(Y+\Delta)^{-1/2}(Y+\Delta)(Y+\Delta)^{-1/2}]\le \text{tr}[\Delta],
\end{split}
\end{equation}
i.e., 
\begin{equation}
\text{tr}[T^\dagger T\tilde{\pi}]\ge 1-\text{tr}[\Delta]=1-\kappa>0
\end{equation}
such that 
\begin{equation}
\pi=T\tilde{\pi}T^\dagger/\text{tr}[T^\dagger T\tilde{\pi}] \le \frac{Y}{\text{tr}[T^\dagger T\tilde{\pi}]}\le \frac{2^{\lambda} }{1-\kappa}\tilde{\rho},
\end{equation}
 i.e., $S_{\text{max}}(\pi\|\tilde{\rho})\le \log(\frac{2^{\lambda} }{1-\kappa})$.
Now let $|\psi\rangle$ be a purification of $\tilde{\pi}$, $\text{tr}_R[|\psi\rangle\langle\psi|]=\tilde{\pi}$, and write $|\psi^\prime\rangle$ for the unnormalized vector $|\psi^\prime\rangle=T\otimes\id|\psi\rangle$. Then $\langle\psi|\psi^\prime\rangle=\langle\psi|T\otimes\id|\psi\rangle
=\text{tr}[\tilde{\pi}T]$ such that, as $\frac{T+T^\dagger}{2}\le \id$ (which follows from $T^\dagger T\le \id$), we have
\begin{equation}
\begin{split}
1-|\langle\psi|\psi^\prime\rangle|&\le 1-\frac{\langle\psi|\psi^\prime\rangle+\langle\psi|\psi^\prime\rangle^*}{2}
=\text{tr}\left[\tilde{\pi}\left(\id-\frac{T+T^\dagger}{2}\right)\right]
\le \text{tr}\left[(Y+\Delta)\left(\id-\frac{T+T^\dagger}{2}\right)\right]\\
&=\text{tr}\left[(Y+\Delta)\right]
-\frac{1}{2}\text{tr}\left[(Y+\Delta)Y^{1/2}(Y+\Delta)^{-1/2}\right]
-\frac{1}{2}\text{tr}\left[(Y+\Delta)(Y+\Delta)^{-1/2}Y^{1/2}\right]\\
&=\text{tr}\left[(Y+\Delta)\right]
-\text{tr}\left[(Y+\Delta)^{1/2}Y^{1/2}\right]\le \text{tr}[\Delta].
\end{split}
\end{equation}
Finally, 
\begin{equation}
\begin{split}
\|\tilde{\pi}-\pi\|_1&=\left\|\text{tr}_R\left[|\psi\rangle\langle\psi|-\frac{|\psi^\prime\rangle\langle\psi^\prime|}{\text{tr}[T^\dagger T\tilde{\pi}]}\right]\right\|_1
\le \left\||\psi\rangle\langle\psi|-\frac{|\psi^\prime\rangle\langle\psi^\prime|}{\text{tr}[T^\dagger T\tilde{\pi}]}\right\|_1\\
&\le 2\sqrt{1-\frac{|\langle\psi|\psi^\prime\rangle|^2}{\text{tr}[T^\dagger T\tilde{\pi}]}}\le
2\sqrt{1-(1-\text{tr}[\Delta])^2}\le \sqrt{8\kappa}.
\end{split}
\end{equation}
\end{proof}

\begin{lemma}  \label{tracenormboundfromxi} For all $\rho_{A_1\cdots A_M} \in {\cal D}((\mathbb{C}^D)^{\otimes M})$   
\begin{equation}
\left \Vert \rho_{A_1\cdots A_M} - \rho_{A_1} \otimes \ldots \otimes \rho_{A_M}  \right \Vert_1 \leq D^2 \sum_{j=2}^M\text{cor}_{\rho_{A_1\cdots A_j}}(A_1 \cdots A_{j-1}, A_{j}).
\end{equation}
\end{lemma}

\begin{proof}
By Lemma 20 of Ref.~\cite{BH12},
for every $j =1,\dots,M$,
\begin{equation}
\left \Vert  \rho_{A_1\cdots A_j} -  \rho_{A_1\cdots A_{j-1}} \otimes \rho_{A_j} \right \Vert_1 \leq D^2 \text{cor}_{\rho_{A_1\cdots A_M}}(A_1 \cdots A_{j-1}, A_{j}). 
\end{equation}
Then  by a telescoping sum and triangle inequality,
\begin{eqnarray}
\left \Vert \rho_{A_1 \cdots A_M}    - \rho_{A_1} \otimes \cdots \otimes \rho_{A_M}     \right \Vert_1 &=&  \Bigl \Vert  \sum_{j=2}^M \left(L_{j} - L_{j-1} \right)  \Bigr \Vert _1      \nonumber \\
&\leq&
 \sum_{j=2}^M \bigl\Vert \bigl(\rho_{A_1 \cdots A_{j}} -  \rho_{A_1 \cdots A_{j-1}} \otimes \rho_{A_j} \bigr)\otimes \rho_{A_{j+1}}\otimes\cdots \otimes\rho_{A_M}\bigr\Vert_1
\nonumber \\
&\leq&
D^2 \sum_{j=2}^M\text{cor}_{\rho_{A_1\cdots A_j}}(A_1 \cdots A_{j-1}, A_{j}),
\end{eqnarray}
with $L_{j} = \rho_{A_1 \cdots A_{j}} \otimes \rho_{A_{j+1}} \otimes \cdots \otimes \rho_{A_M}$.
\end{proof}

\acknowledgements
FB acknowledges EPSRC for financial support. MC acknowledges the EU Integrated Project SIQS and the Alexander von Humboldt foundation for financial support. Part of this work was done while FB was visiting the Simons Institute for the Theory of Computing in the program Quantum Hamiltonian Complexity.

\end{document}